\documentclass[11pt]{article}
\usepackage{libertine}
\usepackage{fullpage}
\usepackage{microtype}

\usepackage{amsmath,amssymb,amsthm,mathtools,thmtools}

\usepackage[breaklinks=true]{hyperref}
\usepackage[svgnames]{xcolor}
\usepackage[capitalise,nameinlink]{cleveref}
\hypersetup{colorlinks={true},linkcolor={DarkBlue},citecolor=[named]{DarkGreen}}

\usepackage{natbib}

\usepackage{algorithm}
\usepackage{algpseudocode}
\usepackage{float}

\newtheorem{theorem}{Theorem}[section]
\newtheorem{lemma}[theorem]{Lemma}
\newtheorem{proposition}[theorem]{Proposition}
\newtheorem{corollary}[theorem]{Corollary}

\usepackage[svgnames]{xcolor}

\usepackage{cleveref}

\newcommand{\voters}{\mathcal{V}} 
\newcommand{\candidates}{\mathcal{C}} 
\newcommand{\ut}{\textsf{u}}
\newcommand{\uprofile}{\mathcal{U}} 
\newcommand{\sprofile}{\mathcal{S}} 
\newcommand{\tr}{\boldsymbol{\tau}} 
\newcommand{\sw}{\operatorname{\normalfont\textsf{SW}}} 
\newcommand{\vsw}{\operatorname{\normalfont\textsf{VSW}}} 
\newcommand{\fd}{\mathrm{Det}} 
\newcommand{\fr}{\mathrm{Rand}} 
\newcommand{\distortion}{\operatorname{\textsf{dist}}} 

\newcommand{\expected}{\mathbb{E}} 
\newcommand{\sstar}{\textsf{s}^*}

\newcommand{\N}{\mathcal{N}}
\newcommand{\R}{\mathbb{R}}
\newcommand{\snorm}[1]{\left\lVert #1\right\rVert_{\N}} 
\newcommand{\cmt}{\normalfont\operatorname{\textsf{W}}} 
\newcommand{\ones}{\mathbf{1}} 

\usepackage{authblk}

\title{\bf The Distortion of Threshold Approval Multiwinner Voting}

\author[1]{Mohamad Latifian}
\affil[1]{University of Edinburgh, UK}
\author[2]{Marzieh Nilipour}
\affil[2]{Sharif University of Technology, Iran}
\author[3]{Alexandros A. Voudouris}
\affil[3]{University of Southern Denmark, Denmark}

\date{}
\begin{document}
\maketitle

\begin{abstract}
	We study multiwinner voting with $n$ voters and $m$ candidates, where the goal is to select a committee of size $k$. Voters have unit-sum cardinal utilities but report only coarse information induced by $t$ utility thresholds. In particular, each voter partitions the candidates into approval sets, with the utilities of all candidates in the same set lying between the same pair of consecutive thresholds. We model the utility of each voter for a committee as a normalized monotone symmetric norm, capturing, among others, additive welfare, best-representative utility, top-$r$ welfare, $\ell_p$ welfare, and normalized non-increasing ordered weighted averages. We design a deterministic rule with distortion at most $2(2m)^{1/t}$ and a randomized rule with distortion at most $2\min\{m^{1/(t+1)},m/k\}$. These guarantees hold for every committee size $k$ and every norm in our framework. Consequently, a single threshold yields distortion $O(m)$ for deterministic rules and $O(\sqrt{m})$ for randomized rules, while $O(\log m)$ thresholds suffice for constant distortion. We complement these results with norm-dependent lower bounds that asymptotically match our upper bounds for constant committee sizes, including the single-winner case. 
\end{abstract}

\section{Introduction}
Multiwinner voting studies how voters select a fixed-size \emph{committee} of candidates. Such problems arise in a wide range of applications, including the election of representatives, shortlisting applicants for a job, and choosing a diverse set of alternatives for recommendation or search \citep{elkind-et-al-2017,lackner-skowron-2023}. A fundamental question in these settings is how much information about the preferences of the voters should be elicited. Exact \emph{cardinal utilities} provide detailed information but may be demanding to report, whereas more commonly used forms of input, such as \emph{ordinal rankings} and \emph{approval votes}, require substantially less effort at the cost of losing information about the strength of preferences. This loss of information may lead to outcomes with significantly less social welfare, a phenomenon commonly quantified via the worst-case notion of \emph{distortion}~\citep{boutilier-et-al-2015,distortion-survey}.

\emph{Threshold approvals} provide an intermediate form of preference elicitation that lies between coarse ordinal information and full cardinal utilities. They have recently been considered within the distortion framework for several related settings, including metric single-winner voting \citep{anshelevich-et-al-2025}, utilitarian matching \citep{latifian-voudouris-2025}, and metric facility-assignment problems \citep{gkatzelis-et-al-2026}. The closest to our setting is the work of \citet{latifian-voudouris-2025}, which considers the same threshold-based information model and the same underlying cardinal utility model, but for matching rather than multiwinner voting. In our setting, given a collection of common utility thresholds, each voter partitions the candidates according to the intervals containing the utilities she assigns to them. With a single threshold, each voter distinguishes between candidates she values above or below that threshold, while additional thresholds allow her to express progressively finer levels of preference intensity.

The multiwinner setting involves an additional modeling question: \emph{how should the utilities that a voter assigns to the selected candidates be combined into her utility for the committee?} The appropriate answer depends on the application at hand. If the committee consists of projects and a voter is interested in all of them, it is natural to sum the utilities she assigns to each of them. If the committee instead consists of representatives or recommendations, the voter may care primarily about only one or a few of the committee members she values most, as in representation-focused objectives related to Chamberlin-Courant \citep{chamberlin-courant-1983}. More generally, committee members may contribute differently depending on how highly the voter values them, for example by considering only the top-$r$ committee members or by assigning decreasing importance to lower-valued members, as in OWA-based approaches to multiwinner voting \citep{elkind-ismaili-2015}. To capture these different applications within a common framework, we define the utility of a voter for a committee through a normalized monotone symmetric norm, which includes additive, best-representative, top-$r$, $\ell_p$-norm, and normalized non-increasing OWA utility. This leads to the central question of our work, namely \emph{the tradeoff between the number of utility thresholds and the achievable distortion}.

\subsection{Our Contribution}
We focus on a multiwinner setting with $n$ voters and $m$ candidates, where each voter has unit-sum utilities over the candidates and the goal is to select a committee of size $k$. Given a collection of utility thresholds, each voter reports a partition of the candidates into approval sets according to the threshold intervals containing the utilities she assigns to them. These threshold approvals are the input to the voting rule. The utility of each voter for a committee is defined by a normalized monotone symmetric norm, and the social welfare is the sum of these utilities. The norm determines how the values she assigns to the selected candidates are combined, while the threshold approvals determine how much information about these values is available to the voting rule.

We first design a deterministic rule that achieves distortion at most $2(2m)^{1/t}$ by using $t$ geometrically spaced thresholds. The rule transforms the threshold approvals into \emph{virtual utilities} that are equal to the lower endpoints of the corresponding threshold intervals, and then selects a committee that maximizes the virtual welfare. We then turn to randomization and design a rule with distortion at most $2\min\{m^{1/(t+1)},m/k\}$. The rule selects, with equal probability, either a committee maximizing virtual welfare or a uniformly random size-$k$ committee. The two branches play complementary roles: the former exploits the information provided by the threshold approvals, while the latter guarantees a baseline level of welfare independently of them. These results are presented in \Cref{sec:upper}.

We complement these upper bounds in \Cref{sec:lower} with lower bounds that hold for arbitrary thresholds. The main idea is to exploit the fact that different utility values within the same threshold interval induce the same threshold approval and are therefore indistinguishable. For deterministic rules, if the last threshold is too large, we can place every candidate below this threshold while concentrating most of the utility on candidates outside the selected committee, leading to unbounded distortion. Otherwise, there must be a sufficiently large multiplicative gap between two consecutive thresholds. We exploit this gap to construct an approval set containing $k+1$ candidates while allowing one of them to have substantially larger utility than the others. For randomized rules, we follow the same idea, but exploit candidates that are selected with sufficiently small probability. The resulting lower bounds capture the dependence on $m$, $k$, and the norm used to define committee utility. In particular, for single-winner voting and, more generally, for committees of constant size, they yield that the optimal distortion is $\Theta(m^{1/t})$ for deterministic rules and $\Theta(m^{1/(t+1)})$ for randomized rules.
In \Cref{sec:best-rep}, we focus on best-representative utility and establish tight deterministic bounds when the committee excludes a constant fraction of the candidates, as well as tight randomized bounds for every committee size.

\subsection{Related Work}
Much of the early literature on distortion focused on single-winner voting rules that have access only to the ordinal preferences of the voters. Two main models have been considered, depending on whether the voters have underlying unit-sum cardinal utilities~\citep{procaccia-rosenschein-2006,boutilier-et-al-2015,ebadian-et-al-2024} or metric costs induced by distances~\citep{gkatzelis-et-al-2020,charikar-et-al-2024}. 
Beyond single-winner voting, the distortion of ordinal rules has been studied for many other social choice problems, including multiwinner voting~\citep{caragiannis-et-al-2017,caragiannis-et-al-2022,babashah-et-al-2025} and one-sided matching~\citep{filos-ratsikas-et-al-2014,amanatidis-et-al-2022,anari-et-al-2023}. 
We refer to \citet{distortion-survey} for a broader overview of the distortion literature.

A more recent line of work studies how distortion improves when rules have access to additional information about the underlying cardinal preferences. One approach is to combine ordinal rankings with a limited number of cardinal queries, such as value or distance queries, leading to explicit tradeoffs between their number and the achievable distortion~\citep{amanatidis-et-al-2021,amanatidis-et-al-2022,amanatidis-et-al-2024,burkhardt-et-al-2024,pulyassary-swamy-2025,fotakis-et-al-2025}. \citet{mandal-et-al-2019,mandal-et-al-2020} studied related tradeoffs between distortion and the communication complexity of voting rules. 
Another related direction considers learning-augmented rules, where ordinal preferences are accompanied by possibly inaccurate predictions about the underlying cardinal information, leading to tradeoffs between consistency and robustness~\citep{berger-et-al-2024,filos-ratsikas-et-al-2026}.

Another way of providing limited cardinal information is through threshold approvals. \citet{benade-et-al-2021} studied threshold approval votes as an elicitation method for participatory budgeting, where projects have costs and the selected outcome is subject to a budget constraint. \citet{bhaskar-et-al-2018} considered threshold approval voting for welfare maximization in multiwinner elections under additive utilities, but each voter was presented with a single independently chosen threshold. In contrast, our model uses multiple common thresholds and allows for more general ways of defining a voter's utility for a committee. As already mentioned, threshold approvals were also used to improve distortion in metric single-winner voting~\citep{anshelevich-et-al-2025} and metric facility-assignment problems~\citep{gkatzelis-et-al-2026}. Finally, \citet{latifian-voudouris-2025} considered multiple common utility thresholds in matching problems and established distortion bounds as a function of the number of thresholds.

Our work is also related to the literature on how voters evaluate committees in multiwinner elections. Representation-focused objectives assign to each voter the utility of her most preferred committee member~\citep{chamberlin-courant-1983}, while additive objectives sum the utilities of all selected candidates. More generally, OWA-based rules allow committee members to contribute differently according to how highly they are valued by a voter~\citep{elkind-ismaili-2015}, and broader classes of committee scoring rules have been studied in~\citep{elkind-et-al-2017,faliszewski-et-al-2016}. Existing work on distortion in multiwinner voting has typically focused on additive committee utility or objectives determined by the $q$-th most-preferred committee member. Our normalized monotone symmetric norm model encompasses these cases, as well as top-$r$, $\ell_p$-norm, and normalized non-increasing OWA utility.


\section{Model}\label{sec:model}
For a positive integer $r$, let $[r]=\{1,\ldots,r\}$. 

\subsection{Utilities and Elicitation Method}
An election instance consists of a set $\voters=[n]$ of $n$ \emph{voters} and a set $\candidates=\{c_1,\ldots,c_m\}$ of $m\ge2$ \emph{candidates}. Voter $i\in\voters$ has a unit-sum utility function $\ut_i:\candidates\to[0,1]$ such that $\sum_{c\in\candidates}\ut_i(c)=1$. 
Together, the utility functions form the \emph{utility profile} $\uprofile \coloneqq(\ut_1,\ldots,\ut_n)$.
We also think of $\ut_i$ as a vector in $\R^m$ whose coordinates are indexed by candidates. 

Fix an integer $t\ge1$ and a common \emph{threshold vector} $\tr=(\tau_1,\ldots,\tau_t)$ with $1=\tau_0>\tau_1>\cdots>\tau_t>0$. Voter $i$ reports $t+1$ disjoint \emph{approval sets}
\begin{align*}
	\sprofile_{i,1}&\coloneqq\{c\in \candidates:1\ge \ut_i(c)\ge\tau_1\},\\
	\sprofile_{i,j}&\coloneqq\{c\in \candidates:\tau_{j-1}>\ut_i(c)\ge\tau_j\}
	&& (2\le j\le t),\\
	\sprofile_{i,t+1}&\coloneqq\{c\in \candidates:\tau_t > \ut_i(c) \ge 0\}.
\end{align*}
Let $\sprofile\coloneqq(\sprofile_{i,j})_{i\in \voters,\,j\in[t+1]}$ be the input profile. We will write $\uprofile\triangleright \sprofile$ when $\uprofile$ satisfies the unit-sum assumption and induces the input profile reported by the voters. 

\subsection{From Candidate Utility to Committee Utility}\label{sec:5norm-def}
For a fixed integer $k\in[m]$, a committee is a set $\cmt\subseteq \candidates$ of size $k$.
Let $\ut_i^{\cmt}$ be the vector obtained from $\ut_i$ by keeping the coordinates in $\cmt$, and setting every other coordinate to zero. More generally, we use $\mathbf{x}^A$ for the restriction of a vector $\mathbf{x}$ to a set $A\subseteq \candidates$. Several familiar ways to measure committee utility can be expressed in terms of $\ut_i^{\cmt}$.
\begin{itemize}
	\item 
	{\bf Best-representative utility.}
	A voter who cares only about her favorite selected candidate has utility
	\[
	\max_{c\in \cmt}\ut_i(c)=\lVert \ut_i^{\cmt}\rVert_\infty,
	\]
	where the maximum is taken over the members of the committee for the individual voter. This corresponds to the cardinal analogue of the Chamberlin-Courant objective, which aggregates such best-representative utilities across voters.
	
	\item 
	{\bf Additive utility.}
	A voter who benefits from all selected candidates has utility
	\[
	\sum_{c\in \cmt}\ut_i(c)=\lVert \ut_i^{\cmt}\rVert_1.
	\]
	Thus, every selected candidate contributes her full utility to the voter. Unlike best-representative utility, this objective takes into account the utilities of all committee members rather than only that of the voter’s favorite one.
	
	\item 
	{\bf Top-$r$ utility.}
	Order the selected utilities as $\ut_i^{(1)}(\cmt)\ge\cdots\ge \ut_i^{(k)}(\cmt)$. For an integer $r\in[m]$, top-$r$ utility gives voter $i$ utility
	\[
	\sum_{j=1}^{\min\{r,k\}}\ut_i^{(j)}(\cmt).
	\]
	Thus, only the voter’s best $r$ committee members contribute to her utility. This recovers best-representative utility when $r=1$ and additive utility when $r\ge k$.
	
	\item 
	{\bf Ordered weighted average utility.}
	For weights $w_1\ge\cdots\ge w_k\ge0$ with $w_1>0$, an ordered weighted average (OWA) assigns voter $i$ utility
	\[
	\sum_{j=1}^k w_j \ut_i^{(j)}(\cmt).
	\]
	Thus, the weight assigned to a selected candidate depends on her position in the ordering of the voter’s selected utilities, rather than on her identity. When placing this objective in our norm model, we normalize the weights so that $w_1=1$, since multiplying all weights by the same positive factor only rescales every committee utility by that factor.
	
	\item 
	{\bf $\ell_p$ utility.}
	For $1\le p<\infty$, the $\ell_p$ objective assigns voter $i$ utility
	\[
	\lVert \ut_i^{\cmt}\rVert_p=\left(\sum_{c\in \cmt}\ut_i(c)^p\right)^{1/p}.
	\]
	These objectives interpolate between additive utility at $p=1$ and best-representative utility in the limit as $p$ tends to infinity. As $p$ increases, relatively more weight is placed on the voter’s highest utilities among the selected candidates.
\end{itemize}

\subsection{Monotone Symmetric Norms}
To formalize how a voter evaluates a committee of candidates, we adopt a norm $\snorm{\cdot}$ on $\R^m$ that satisfies a set of basic normative axioms:
\begin{itemize}
	\item {\bf Zero benefit.}
	A voter receives zero committee utility exactly when every selected candidate has zero utility for her:
	\[
	\snorm{\mathbf{x}}=0 \iff \mathbf{x}=0.
	\]
	This is precisely the \emph{positive definiteness} property of a norm.
	
	\item {\bf Monotonicity.}
	For all non-negative vectors $\mathbf{x},\mathbf{y}$:
	\[
	0\le \mathbf{x}\le \mathbf{y}
	\implies
	\snorm{\mathbf{x}}\le \snorm{\mathbf{y}},
	\]
	where the inequalities between vectors are coordinate-wise. Thus, if the utilities of the selected candidates weakly increase for a voter, her committee utility cannot decrease.
	
	\item {\bf Symmetry.}
	For any vector $\mathbf{x}\in\R^m$ and any $m\times m$ permutation matrix $P$,
	\[
	\snorm{P\mathbf{x}}=\snorm{\mathbf{x}}.
	\]
	Thus, permuting the utilities of the selected candidates does not change the voter's committee utility. In other words, the aggregator depends only on the multiset of utilities of the selected candidates, and not on their ordering.
	
	\item {\bf Positive homogeneity.}
	Multiplying every selected utility by a scalar $\alpha\ge0$ multiplies the committee utility by the same factor:
	\[
	\snorm{\alpha \mathbf{x}}=\alpha\snorm{\mathbf{x}}.
	\]
	
	\item {\bf Triangle inequality.}
	For any two non-negative utility vectors $\mathbf{x}$ and $\mathbf{y}$,
	\[
	\snorm{\mathbf{x}+\mathbf{y}}
	\le
	\snorm{\mathbf{x}}+\snorm{\mathbf{y}}.
	\]
	Thus, the committee utility obtained from the coordinate-wise sum of two utility vectors cannot exceed the sum of their separate committee utilities.
	
	\item {\bf Normalization.}
	We set $\snorm{\mathbf{e}_c}=1$ for every candidate $c$. By positive homogeneity,
	\[
	\snorm{\alpha \mathbf{e}_c}=\alpha.
	\]
	Thus, a singleton committee containing candidate $c$ gives voter $i$ committee utility exactly $\ut_i(c)$. This places different norms on a common scale.
\end{itemize}

All utility definitions presented in \Cref{sec:5norm-def} satisfy these properties. For completeness, the top-$r$ norm of an arbitrary vector $\mathbf{x}\in\R^m$ is defined as the sum of its $r$ largest absolute coordinates, or equivalently,
\[
\max_{|A|=r}\sum_{c\in A}|\mathbf{x}_c|.
\]
This variational representation immediately implies the triangle inequality, while the remaining norm properties are straightforward to verify.

To express an OWA as a norm on $\R^m$, we can extend the weights by zeros beyond $k$ and set $w_{m+1}\coloneqq0$. Writing $|\mathbf{x}|_{(1)}\ge\cdots\ge|\mathbf{x}|_{(m)}$ for the sorted absolute coordinates, we have
\[
\lVert \mathbf{x}\rVert_w\coloneqq\sum_{j=1}^m w_j|\mathbf{x}|_{(j)}
=\sum_{r=1}^m(w_r-w_{r+1})\lVert \mathbf{x}\rVert_{\mathrm{top}\text{-}r}.
\]
Since the weights are non-increasing, this is a non-negative combination of top-$r$ norms, with at least one positive coefficient. It is therefore a monotone symmetric norm, normalized when $w_1=1$.

\subsection{Committee Welfare under Monotone Symmetric Norms} 
For a fixed normalized monotone symmetric norm, we define voter $i$'s committee utility as
\[
U_i^\N(\cmt)\coloneqq\snorm{\ut_i^{\cmt}}.
\]
Then, the social welfare of a committee $\cmt$ is the total voter utility:
\[
\sw_\N(\cmt,\uprofile)\coloneqq\sum_{i\in \voters}U_i^\N(\cmt).
\]
When clear from context, we omit $\uprofile$ and simply write $\sw_\N(\cmt)$. 
For any size-$k$ committee $\cmt$, define
\[
L_{\N,k}\coloneqq\snorm{\ones_{\cmt}},
\]
where $\ones_{\cmt}\in\{0,1\}^m$ is the indicator vector of $\cmt$, with value one on the coordinates corresponding to candidates in $\cmt$ and zero elsewhere. Thus, $L_{\N,k}$ is the norm of a vector with value one on exactly $k$ coordinates. By symmetry, it does not depend on which $k$ coordinates are used. The following bounds place this quantity between its values under best-representative and additive utility.

\begin{proposition}\label{prop:LNK}
	For every normalized monotone symmetric norm and every $k\in[m]$,
	$1\le L_{\N,k}\le k$.
\end{proposition}

\begin{proof}
	Fix any candidate $c\in\cmt$. Since $\ones_{\cmt}\ge \mathbf{e}_c$ coordinate-wise, monotonicity and normalization yield
	\[
	L_{\N,k} =\snorm{\ones_{\cmt}} \ge \snorm{\mathbf{e}_c} =1.
	\]
	For the upper bound, observe that
	\[
	\ones_{\cmt}=\sum_{c\in\cmt}\mathbf{e}_c.
	\]
	Hence, by the triangle inequality and normalization,
	\[
	L_{\N,k} =\snorm{\sum_{c\in\cmt}\mathbf{e}_c} \le \sum_{c\in\cmt}\snorm{\mathbf{e}_c} =k. \qedhere
	\]
\end{proof}

\noindent
For the preceding utility objectives,
$
L_{\ell_1,k}=k,
L_{\ell_p,k}=k^{1/p},
L_{\ell_\infty,k}=1, 
L_{\mathrm{top}\text{-}r,k}=\min\{r,k\}
$.
For normalized OWA weights, we have $L_{w,k}=\sum_{j=1}^k w_j$. These quantities will appear throughout the analysis in the next section, but eventually cancel out and do not appear in the final upper bounds.

\subsection{Committee Rules and Distortion}
A deterministic size-$k$ committee rule fixes a threshold vector and maps every feasible input profile to a committee, while a randomized rule fixes a threshold vector and maps each profile to a probability distribution over size-$k$ committees. The threshold vector is chosen before the voters submit their reports, and the rule has access only to these reports, the public parameters, and the chosen norm. We impose no neutrality or symmetry requirement on the rule.

For a utility profile $\uprofile$, let
$s_\N^*(\uprofile)\coloneqq
\max_{\cmt\subseteq \candidates,\,|\cmt|=k}
\sw_\N(\cmt,\uprofile)$
denote the optimal norm-based social welfare, realized by an optimal committee $\cmt^*$. 
The \emph{distortion} of a rule $f$ using thresholds $\tr$ is
\[
\distortion_\N(f,\tr)\coloneqq
\sup_\sprofile \sup_{\uprofile\triangleright \sprofile}
\frac{s_\N^*(\uprofile)}
{\displaystyle \expected_{\cmt\sim f(\sprofile)}[\sw_\N(\cmt,\uprofile)]}.
\]
The first supremum ranges over feasible input profiles with any number of voters, while the second ranges over utility profiles consistent with the given input profiles. For a deterministic rule, the denominator is simply the welfare of the selected committee. If the numerator is positive and the denominator is zero, the distortion is infinite.


\section{Upper Bounds}\label{sec:upper}
We now present our deterministic and randomized rules and bound their distortion. Both are based on virtual utilities obtained by replacing each reported utility interval with its lower endpoint. We first establish some useful properties of these virtual utilities for the geometric threshold family used throughout this section, together with a baseline welfare bound that will be used in the analysis of the rules.

\subsection{Virtual Utilities and Preliminary Bounds}
Fix $\delta>1$ and consider the thresholds $\tau_j\coloneqq\delta^{-j}$ for $j\in[t]$. 
Given an input profile $\sprofile$, we define the \emph{virtual utility} of voter $i$ for candidate $c$ by
\[
\ut'_i(c)\coloneqq
\begin{cases}
	\tau_j,&c\in \sprofile_{i,j},\ j\in[t];\\
	0,&c\in \sprofile_{i,t+1}.
\end{cases}
\]
Since each virtual utility is the lower endpoint of the corresponding reported utility interval, we have $\ut'_i(c)\le \ut_i(c)$ for every $c\in\candidates$.
Moreover, for any candidate $c \in \sprofile_{i,j}$ with $j \in [t]$, and hence $c\notin\sprofile_{i,t+1}$, we have that $\ut_i(c) \le \tau_{j-1} = \delta \tau_j = \delta \ut'_i(c)$. Consequently,
\[
\ut'_i(c)\ge\frac{\ut_i(c)}\delta.
\]
Thus, virtual utility never overestimates true utility, and any true utility of at least $\tau_t$ is underestimated by a factor of at most $\delta$. We emphasize that the virtual utilities need not satisfy the unit-sum assumption: they are quantities computed from the input profiles and do not necessarily form a feasible utility profile. 
We can now define the {\em virtual welfare} of a committee $\cmt$ as
\[
\vsw_\N(\cmt)\coloneqq\sum_{i\in \voters}\snorm{(\ut'_i)^{\cmt}}.
\]

Since every virtual utility is at most its corresponding true utility, monotonicity of the norm immediately gives the following relation.
\begin{lemma}\label{lem:under}
	For every size-$k$ committee $\cmt$, $\vsw_\N(\cmt)\le\sw_\N(\cmt)$.
\end{lemma}

We next lower-bound the welfare lost by replacing true utilities with their lower endpoints; this bound will later be applied to the optimal committee. For candidates outside the last approval set, virtual utility underestimates true utility by at most a factor $\delta$. Only utilities below $\tau_t$ disappear completely, and for each voter the norm of these discarded utilities is at most $\tau_t L_{\N,k}$.

\begin{lemma}\label{lem:loss}
	For every size-$k$ committee $\cmt$,
	\[
	\vsw_\N(\cmt)\ge\frac1\delta\bigl(\sw_\N(\cmt)-n\tau_tL_{\N,k}\bigr).
	\]
\end{lemma}

\begin{proof}
	Fix a voter $i$, and let
	\[
	A_i\coloneqq\cmt\cap\bigcup_{j=1}^t \sprofile_{i,j},
	\qquad B_i\coloneqq\cmt\cap \sprofile_{i,t+1}.
	\]
	For any $c\in A_i$, we have $\ut'_i(c)\ge \frac1\delta \ut_i(c)$. Since the virtual utilities vanish on $B_i$, monotonicity and homogeneity give
	\[
	\snorm{(\ut'_i)^{\cmt}}=\snorm{(\ut'_i)^{A_i}}
	\ge\frac1\delta\snorm{(\ut_i)^{A_i}}.
	\]
	Every coordinate of $(\ut_i)^{B_i}$ is below $\tau_t$, so $(\ut_i)^{B_i}\le\tau_t \ones_{\cmt}$ coordinate-wise.
	Monotonicity and homogeneity then imply $\snorm{(\ut_i)^{B_i}} \le \tau_t \snorm{\ones_{\cmt}} = \tau_t L_{\N,k}$. 
	Since $(\ut_i)^{\cmt}=(\ut_i)^{A_i}+(\ut_i)^{B_i}$, the triangle inequality implies
	\[
	\snorm{(\ut_i)^{A_i}}\ge\snorm{(\ut_i)^{\cmt}}-\snorm{(\ut_i)^{B_i}}
	\ge\snorm{(\ut_i)^{\cmt}}-\tau_tL_{\N,k}.
	\]
	Combining the inequalities yields
	\[
	\snorm{(\ut'_i)^{\cmt}}\ge\frac1\delta
	\biggl(\snorm{(\ut_i)^{\cmt}}-\tau_tL_{\N,k}\biggr).
	\]
	Summing over all voters proves the desired result.
\end{proof}

We next establish a baseline welfare bound that will be useful in both the deterministic and randomized analyses. We begin by comparing the norm of a non-negative vector with its $\ell_1$ norm. The following inequality follows by averaging permutations of a vector within $k$ coordinates and distributing its total sum equally across those coordinates.

\begin{lemma}\label{lem:average}
	For every non-negative vector $\mathbf{x} \in \R^m$ supported on at most $k$ coordinates,
	\[
	\snorm{\mathbf{x}}\ge\frac{L_{\N,k}}k\lVert \mathbf{x}\rVert_1.
	\]
\end{lemma}

\begin{proof}
	Choose a size-$k$ set $\cmt$ containing the support of $\mathbf{x}$, and consider all $k!$ permutations of the coordinates of $\mathbf{x}$ within $\cmt$. By symmetry, every resulting vector has norm $\snorm{\mathbf{x}}$, and their average is
	\[
	\frac{\lVert\mathbf{x}\rVert_1}{k}\ones_{\cmt}.
	\]
	By the triangle inequality and positive homogeneity, the norm of this average is at most $\snorm{\mathbf{x}}$. Hence,
	\[
	\frac{\lVert\mathbf{x}\rVert_1}{k}\snorm{\ones_{\cmt}}
	\le \snorm{\mathbf{x}}.
	\]
	Since $\snorm{\ones_{\cmt}}=L_{\N,k}$, the result follows.
\end{proof}

Every candidate appears in a uniformly random size-$k$ committee with probability $\frac{k}{m}$. Together with \Cref{lem:average} and the unit-sum assumption, this yields the following report-independent baseline welfare guarantee.

\begin{corollary}\label{cor:random}
If $\cmt_r$ is a uniformly random size-$k$ committee, then
\[
    \expected[\sw_\N(\cmt_r)]\ge\frac{nL_{\N,k}}m.
\]
\end{corollary}

\begin{proof}
	Applying \Cref{lem:average} to each utility vector $(\ut_i)^{\cmt_r}$, and using linearity of expectation and the unit-sum assumption, we obtain
	\begin{align*}
		\expected[\sw_\N(\cmt_r)]
		&\ge \frac{L_{\N,k}}{k}\sum_{i\in\voters}
		\expected\left[\left\|(\ut_i)^{\cmt_r}\right\|_1\right]\\
		&=\frac{L_{\N,k}}{k}\sum_{i\in\voters}\sum_{c\in\candidates}
		\Pr[c\in\cmt_r]\ut_i(c)\\
		&=\frac{L_{\N,k}}{k}\sum_{i\in\voters}\sum_{c\in\candidates}
		\frac{k}{m}\ut_i(c)
		=\frac{nL_{\N,k}}m. \qedhere
	\end{align*}
\end{proof}

Since the optimal committee has welfare at least the average welfare over all size-$k$ committees, \Cref{cor:random} implies the following lower bound on the optimal welfare:
\begin{equation}\label{eq:baseline}
	s_\N^*\ge\frac{nL_{\N,k}}m.
\end{equation}
Since $L_{\N,k}\ge1$ by \Cref{prop:LNK}, we have that $s_\N^*\ge n/m>0$. Hence, the optimal welfare is strictly positive for every utility profile.

The preceding corollary provides a lower bound on the expected welfare of a uniformly random committee that is independent of the utility profile. We next establish a complementary bound in terms of the optimal welfare.

\begin{lemma}\label{lem:uniform-to-optimum}
If $\cmt_r$ is a uniformly random size-$k$ committee, then
\[
    \expected[\sw_\N(\cmt_r)]\ge\frac{k}{m}s_\N^*.
\]
\end{lemma}

\begin{proof}
Let $\cmt^*$ be an optimal committee, and fix a voter $i$. Since every candidate belongs to a uniformly random size-$k$ committee with probability $k/m$, we have
\[
    \expected\left[(\ut_i)^{\cmt_r}\right]=\frac{k}{m}\ut_i.
\]
By the triangle inequality and positive homogeneity,
\[
    \expected\left[\snorm{(\ut_i)^{\cmt_r}}\right]
    \ge \snorm{\expected\left[(\ut_i)^{\cmt_r}\right]}
    =\frac{k}{m}\snorm{\ut_i}.
\]
Moreover, since $(\ut_i)^{\cmt^*}\le\ut_i$ coordinate-wise, monotonicity implies $\snorm{\ut_i}\ge\snorm{(\ut_i)^{\cmt^*}}$.
Therefore,
\[
    \expected\left[\snorm{(\ut_i)^{\cmt_r}}\right] \ge\frac{k}{m}\snorm{(\ut_i)^{\cmt^*}}.
\]
Summing over all voters, we obtain
\[
    \expected[\sw_\N(\cmt_r)]
    \ge\frac{k}{m}\sw_\N(\cmt^*)
    =\frac{k}{m}s_\N^*. \qedhere
\]
\end{proof}

\subsection{Deterministic Rule}\label{sec:det}
Our deterministic rule selects a committee that maximizes social welfare with respect to the virtual utilities defined above, breaking ties arbitrarily but consistently.
\begin{algorithm}[H]
	\caption{$\fd_\delta$}
	\label{alg:det_rule}
	\begin{algorithmic}[1]
		\Require Input profile $\sprofile$, parameter $\delta>1$
		\State Set thresholds $\tau_j\gets\delta^{-j}$ for every $j\in[t]$
		\State Compute the virtual utilities $\ut_i'(c)$ from $\sprofile$ using these thresholds
		\State \Return a size-$k$ committee in ${\arg\max}_{\cmt'\subseteq\candidates,\,|\cmt'|=k}\vsw_\N(\cmt')$
	\end{algorithmic}
\end{algorithm}
\noindent 

We now bound the distortion. The parameters $k$ and $\N$ are fixed throughout, and are therefore omitted from notation. We consider the rule $\fd_\delta$ with final threshold $\tau_t=1/(2m)$. Equivalently, we choose $\delta$ so that $\delta^{-t}=1/(2m)$. By Equation~\eqref{eq:baseline}, the welfare lost by discarding utilities below the final threshold is then at most half of the optimal welfare.

\begin{theorem}\label{thm:det}
	For every $t\ge1$, $k\in[m]$, and every normalized monotone symmetric norm, let $\delta=(2m)^{1/t}$. Then, the distortion of $\fd_\delta$ is at most $2(2m)^{1/t}$.
\end{theorem}

\begin{proof}
	Fix an input profile $\sprofile$ and a utility profile $\uprofile\triangleright\sprofile$. Since $\tau_t=\delta^{-t}=1/(2m)$, Equation~\eqref{eq:baseline} gives
	\[
	n\tau_tL_{\N,k} = \frac{nL_{\N,k}}{2m}
	\le \frac{s_\N^*}{2}.
	\]
	Let $\cmt=\fd_\delta(\sprofile)$. By \Cref{lem:under,lem:loss},
	\[
	\sw_\N(\cmt)
	\ge \vsw_\N(\cmt)
	\ge \vsw_\N(\cmt^*)
	\ge \frac{1}{\delta}\bigl(s_\N^*-n\tau_tL_{\N,k}\bigr)
	\ge \frac{s_\N^*}{2\delta}.
	\]
	Hence, the distortion is at most $2\delta=2(2m)^{1/t}$.
\end{proof}

\subsection{Randomized Rule}\label{sec:rand}
Randomization improves the dependence on the number of candidates. The key idea is to mix the deterministic rule presented above, which exploits the reported information, with a uniformly random size-$k$ committee, which provides a report-independent baseline welfare guarantee.

\begin{algorithm}[H]
	\caption{$\fr_\delta$}
	\label{alg:rand_rule}
	\begin{algorithmic}[1]
		\Require Input profile $\sprofile$, parameter $\delta>1$
		\State Set thresholds $\tau_j\gets\delta^{-j}$ for every $j\in[t]$
		\State With probability $1/2$, choose a size-$k$ committee $\cmt$ uniformly at random
		\State With probability $1/2$, set $\cmt\gets\fd_\delta(\sprofile)$
		\State \Return $\cmt$
	\end{algorithmic}
\end{algorithm}

To obtain the improved distortion bound, we choose a larger final threshold than in the deterministic analysis. The two branches complement each other: when the optimal welfare is small, the uniformly random committee provides a sufficient baseline guarantee, while when it is large, replacing true utilities by their lower endpoints loses only a bounded fraction of the optimal welfare.

\begin{theorem}\label{thm:rand}
For every $t\ge1$, $k\in[m]$, and every normalized monotone symmetric norm, let $\delta=m^{1/(t+1)}$.
Then, the distortion of $\fr_\delta$ is at most
\[
	2\min\left\{m^{1/(t+1)},\frac{m}{k}\right\}.
\]
\end{theorem}

\begin{proof}
    Fix an input profile $\sprofile$ and a utility profile
    $\uprofile\triangleright\sprofile$. Let
    $\cmt_d=\fd_\delta(\sprofile)$, let $\cmt_r$ be a uniformly
    random size-$k$ committee, and let $\cmt=\fr_\delta(\sprofile)$.
    As in the proof of \Cref{thm:det}, we have
    \[
        \sw_\N(\cmt_d)\ge
        \frac{1}{\delta}(s_\N^*-n\tau_tL_{\N,k}).
    \]
    Thus, by \Cref{cor:random},
    \[
    \begin{aligned}
        \expected[\sw_\N(\cmt)]
        &=\frac12\sw_\N(\cmt_d)
        +\frac12\expected[\sw_\N(\cmt_r)]\\
        &\ge\frac{1}{2\delta}
        (s_\N^*-n\tau_tL_{\N,k})
        +\frac{nL_{\N,k}}{2m}\\
        &=\frac{s_\N^*}{2\delta},
    \end{aligned}
    \]
    where the last equality follows from $\tau_t/\delta=\delta^{-(t+1)}=1/m$, which makes the two terms involving $nL_{\N,k}$ cancel.
    
    Moreover, by \Cref{lem:uniform-to-optimum},
    \[
        \expected[\sw_\N(\cmt)]
        \ge\frac12\expected[\sw_\N(\cmt_r)]
        \ge\frac{k}{2m}s_\N^*.
    \]
    Combining the two bounds, we have that the distortion of $\fr_\delta$ is at most $2\min\{m^{1/(t+1)},m/k\}$.
\end{proof}

Note that \Cref{thm:det,thm:rand} hold for all utility functions introduced in \Cref{sec:5norm-def}, including additive, best-representative, top-$r$, $\ell_p$, and normalized non-increasing OWA utility. The chosen norm affects which committee maximizes virtual welfare, and in some cases finding such a committee may be computationally hard, but the distortion analysis remains unchanged.

\subsection{Consequences and Special Cases}
The number of thresholds captures the tradeoff between the amount of elicited information and distortion. At one extreme, a single threshold yields deterministic distortion at most $4m$ and randomized distortion at most $2\sqrt m$. At the other, $t=\Theta(\log m)$ thresholds already suffice to achieve constant distortion for both rules.

\begin{corollary}
	For every normalized monotone symmetric norm, one threshold suffices for deterministic distortion at most $4m$ and randomized distortion at most $2\sqrt m$. Choosing $t=\lceil\log_2m\rceil$ gives deterministic distortion at most $8$ and randomized distortion at most $4$.
\end{corollary}

The single-threshold bounds follow directly from \Cref{thm:det,thm:rand} with $t=1$. If $t=\lceil\log_2m\rceil$, then $m\le2^t$, so $(2m)^{1/t}\le4$, giving deterministic distortion at most $8$, while $m^{1/(t+1)}\le2$, giving randomized distortion at most $4$. More generally, $t=\Theta(\log m)$ yields constant distortion for both rules.

\section{Lower Bounds}\label{sec:lower}
We now prove lower bounds for arbitrary committee sizes. The deterministic and randomized arguments are both based on the same dichotomy. If the final threshold is too large, the last approval set can hide substantial utility. Otherwise, the threshold sequence spans a sufficiently large multiplicative range that some pair of consecutive thresholds must contain a large gap. 

\subsection{Deterministic Rules}

{\bf The last-threshold obstruction.}
If every voter assigns utility below the last threshold to every candidate, then all candidates are indistinguishable in the reports, and a deterministic rule may select a committee with arbitrarily small welfare. The following construction makes this explicit.

\begin{lemma}\label{lem:unbounded}
	For $1\le k<m$ and any normalized monotone symmetric norm $\snorm{\cdot}$, if $\tau_t>1/(m-k)$, then every deterministic size-$k$ committee rule using $\tr$ has unbounded distortion.
\end{lemma}

\begin{proof}
	Consider the input profile in which $\sprofile_{i,t+1}=\candidates$ for every voter $i$. 
	This profile is feasible: since $\tau_t>1/(m-k)>1/m$, the uniform utility vector assigns utility below $\tau_t$ to every candidate. Let $\cmt$ be the committee selected by the rule.
	Choose
	\[
	0<\varepsilon<\min\left\{\frac1{m-k},\frac{k\tau_t}{m-k}\right\},
	\]
	and assign every voter utilities
	\[
	\ut_i(c)=
	\begin{cases}
		\dfrac1{m-k}-\varepsilon,&\text{ if }  c\notin \cmt,\\[10pt]
		\dfrac{m-k}{k}\varepsilon,&\text{ if }  c\in \cmt.
	\end{cases}
	\]
	These utilities are non-negative, and satisfy the unit-sum assumption because
	\[
	(m-k)\left(\frac1{m-k}-\varepsilon\right)
	+k\left(\frac{m-k}{k}\varepsilon\right)=1.
	\]
	Moreover, every utility is below $\tau_t$: for candidates outside $\cmt$ this follows from $\tau_t>1/(m-k)$, while for candidates in $\cmt$ it follows from the choice of $\varepsilon$. Hence, this utility profile induces the specified input profile.
	
	The welfare of the selected committee is
	\[
	\sw_\N(\cmt)=nL_{\N,k}\frac{m-k}{k}\varepsilon.
	\]
	Now choose any candidate $c\notin \cmt$ and any size-$k$ committee $\widehat \cmt$ containing $c$. By monotonicity and normalization, the optimal welfare is
	\[
	s_\N^*\ge\sw_\N(\widehat \cmt)\ge
	n\left(\frac1{m-k}-\varepsilon\right).
	\]
	Consequently, the distortion is at least
	\[
	\frac{\frac1{m-k}-\varepsilon}
	{L_{\N,k}\frac{m-k}{k}\varepsilon},
	\]
	which tends to infinity as $\varepsilon > 0$ tends to zero.
\end{proof}

\medskip
\noindent 
{\bf Large gaps between consecutive thresholds.}
We now consider the remaining case where $\tau_t\le 1/(m-k)$.
Set
\begin{align*}
	\sigma_j\coloneqq\max\left\{\tau_j,\frac1{m-k}\right\}
	\quad\text{for }j=0,\ldots,t
\end{align*}
and 
\begin{align*}
	R_j\coloneqq\frac{\sigma_{j-1}}{\sigma_j}
	\quad\text{for }j\in[t].
\end{align*}
The values $\sigma_j$ are auxiliary and do not modify the thresholds used by the rule. In the case $\tau_t\le 1/(m-k)$, we have $\sigma_0=1$ and $\sigma_t=1/(m-k)$, and therefore
\[
\prod_{j=1}^t R_j = \prod_{j=1}^t \frac{\sigma_{j-1}}{\sigma_j} = \frac{\sigma_0}{\sigma_t} = m-k.
\]
A large value of $R_j$ identifies an approval set in which utilities can vary substantially without changing the report. For committees, however, the rule may select up to $k$ candidates from this set. We therefore choose the high utility so that, under the unit-sum assumption, at least $k$ additional report-indistinguishable candidates remain.

\begin{lemma}\label{lem:gap-committee}
	For $1\le k<m$ and any normalized monotone symmetric norm $\snorm{\cdot}$, fix a threshold vector $\tr$ and $j\in[t]$. If $R_j\ge k+1$, 
	then the distortion of every deterministic size-$k$ committee rule $f$ using $\tr$ is
	\[
	\distortion_\N(f,\tr) \ge \frac{R_j-k}{L_{\N,k}}.
	\]
\end{lemma}

\begin{proof}
	Fix $j\in[t]$ with $R_j\ge k+1$. 
	Since $\tau_{j-1}\le1$ and $\sigma_j\ge1/(m-k)$,
	\[
	R_j=\frac{\sigma_{j-1}}{\sigma_j}\le m-k.
	\]
	Hence $m-k\ge k+1$ or, equivalently, $m \geq 2k+1$. 
	
	Since $R_j>1$, we have $\sigma_{j-1}>\sigma_j\ge1/(m-k)$, and hence $\sigma_{j-1}=\tau_{j-1}$.
	Therefore,
	\[
	R_j=\frac{\tau_{j-1}}{\sigma_j}.
	\]
	This directly implies that 
	$$(R_j-k)\sigma_j = \tau_{j-1}-k\sigma_j < \tau_{j-1}.$$
	Since $R_j\ge k+1$, we also have that 
	$$(R_j-k)\sigma_j\ge\sigma_j.$$ 
	Using the fact that $\sigma_j\ge\tau_j$ and $\sigma_j<\tau_{j-1}$, we obtain
	\[
	\tau_j \le \sigma_j \le (R_j-k)\sigma_j < \tau_{j-1}.
	\]
	Therefore, both utility values $\sigma_j$ and $(R_j-k)\sigma_j$ lie in the interval corresponding to approval set $j$.
	
	Choose a set $A\subseteq\candidates$ of candidates with $|A|=k+1$. 
	Consider the input profile $\sprofile$ in which, for every voter, all candidates in $A$ are reported in approval set $j$, while every candidate outside $A$ is reported in the approval set corresponding to the utility value
	\[
	\frac{1-\tau_{j-1}}{m-k-1}.
	\]
	Let $\cmt=f(\sprofile)$. Since $|A|=k+1$, there exists a candidate $c^*\in A\setminus\cmt$.
	Now assign every voter utilities
	\[
	\ut_i(c)=
	\begin{cases}
		(R_j-k)\sigma_j,&\text{ if } c=c^*,\\
		\sigma_j,& \text{ if } c\in A\setminus\{c^*\},\\
		\dfrac{1-\tau_{j-1}}{m-k-1},& \text{ if } c\notin A.
	\end{cases}
	\]
	These utilities satisfy the unit-sum assumption because
	\[
	(R_j-k)\sigma_j +k\sigma_j +(m-k-1)\frac{1-\tau_{j-1}}{m-k-1}
	= R_j\sigma_j+1-\tau_{j-1} = 1.
	\]
	Furthermore,
	\[
	\frac{1-\tau_{j-1}}{m-k-1} \le \sigma_j,
	\]
	since
	\[
	1-\tau_{j-1} = 1-R_j\sigma_j \le (m-k-1)\sigma_j,
	\]
	where the last inequality follows from $(m-k)\sigma_j\ge1$ and $R_j\ge1$. Hence, every candidate other than $c^*$ has utility at most $\sigma_j$. The constructed utility profile also induces exactly $\sprofile$, since changing the utility of $c^*$ from $\sigma_j$ to $(R_j-k)\sigma_j$ leaves it in approval set $j$.
	
	Since $c^*\notin\cmt$, every candidate selected by $\cmt$ has utility at most $\sigma_j$ for every voter. Hence, coordinate-wise,
	$
	\ut_i^{\cmt}\le\sigma_j\ones_{\cmt}.
	$
	By monotonicity and positive homogeneity, we obtain
	\[
	\sw_\N(\cmt)\le n\sigma_jL_{\N,k}.
	\]
	On the other hand, any size-$k$ committee $\widehat\cmt$ containing $c^*$ satisfies
	$
	\ut_i^{\widehat\cmt} \ge (R_j-k)\sigma_j\mathbf{e}_{c^*}.
	$
	By monotonicity and normalization, we have
	\[
	s_\N^* \ge \sw_\N(\widehat\cmt) \ge n(R_j-k)\sigma_j.
	\]
	Consequently,
	\[
	\distortion_\N(f,\tr)
	\ge \frac{n(R_j-k)\sigma_j}{n\sigma_jL_{\N,k}}
	= \frac{R_j-k}{L_{\N,k}}.
	\]
	The proof is now complete.
\end{proof}

We can now combine the last-threshold obstruction with the large-gap bound to obtain the deterministic lower bound.

\begin{theorem}\label{thm:det-committee-lower}
	For $1\le k<m$, $t\ge1$, and any normalized monotone symmetric norm $\snorm{\cdot}$, the distortion of every deterministic size-$k$ committee rule $f$ using an arbitrary threshold vector $\tr$ is
	\[
	\distortion_\N(f,\tr) \ge \max\left\{1,\, \frac{(m-k)^{1/t}-k}{L_{\N,k}} \right\}.
	\]
\end{theorem}

\begin{proof}
	If $\tau_t>1/(m-k)$, \Cref{lem:unbounded} gives unbounded distortion, so it suffices to consider the case $\tau_t\le1/(m-k)$.
	
	Recall that $\sigma_j=\max\{\tau_j,1/(m-k)\}$ and $R_j=\sigma_{j-1}/\sigma_j$. Since $\tau_t\le1/(m-k)$, we have $\sigma_0=1$ and $\sigma_t=1/(m-k)$, and therefore
	\[
	\prod_{j=1}^t R_j = \frac{\sigma_0}{\sigma_t} = m-k.
	\]
	Hence, there exists $j^*\in[t]$ such that $R_{j^*}\ge(m-k)^{1/t}$.
	We now use \Cref{lem:gap-committee}, which gives a lower bound whenever $R_{j^*}\ge k+1$. 
	\begin{itemize}
		\item 
		If $(m-k)^{1/t}\ge k+1$, then indeed $R_{j^*}\ge k+1$, and thus
		\[
		\distortion_\N(f,\tr) \ge \frac{R_{j^*}-k}{L_{\N,k}} \ge \frac{(m-k)^{1/t}-k}{L_{\N,k}}.
		\]
		
		\item If $(m-k)^{1/t}<k+1$, then
		\[
		\frac{(m-k)^{1/t}-k}{L_{\N,k}} < \frac1{L_{\N,k}} \le 1.
		\]
	\end{itemize}
	Combining the two cases proves the claim.
\end{proof}

To put the lower bound into perspective and compare it with the upper bound from the previous section, we next specialize \Cref{thm:det-committee-lower} to the utility functions introduced earlier and identify cases where the resulting bounds are tight up to constant factors.
Substituting $L_{\ell_\infty,k}=1$, $L_{\ell_1,k}=k$ and $L_{\ell_p,k}=k^{1/p}$ into \Cref{thm:det-committee-lower} gives the following bounds.

\begin{corollary}\label{cor:det-committee-lower-norms}
	Let $1\le k<m$ and $t\ge1$. The distortion of every deterministic size-$k$ committee rule $f$ using an arbitrary threshold vector $\tr$ satisfies the following bounds:
	\begin{enumerate}
		\item \textbf{Best-representative utility:}
		\[
		\distortion_{\ell_\infty}(f,\tr)
		\ge
		\max\left\{1,(m-k)^{1/t}-k\right\}.
		\]
		
		\item \textbf{Additive utility:}
		\[
		\distortion_{\ell_1}(f,\tr)
		\ge
		\max\left\{
		1,
		\frac{(m-k)^{1/t}-k}{k}
		\right\}.
		\]
		
		\item \textbf{$\ell_p$ utility with $p \geq 2$:}
		\[
		\distortion_{\ell_p}(f,\tr)
		\ge
		\max\left\{
		1,
		\frac{(m-k)^{1/t}-k}{k^{1/p}}
		\right\}.
		\]
	\end{enumerate}
\end{corollary}

In the single-winner case $k=1$, and more generally for any constant committee size, the lower bound of \Cref{thm:det-committee-lower} matches the upper bound of \Cref{thm:det} up to a constant factor. Hence, the optimal deterministic distortion is $\Theta(m^{1/t})$ in these cases.

\begin{corollary}\label{cor:constant-k-tight}
	For any fixed $k=O(1)$ and normalized monotone symmetric norm, the optimal deterministic distortion with $t$ thresholds is $\Theta(m^{1/t})$.
\end{corollary}

\begin{proof}
	For any fixed $k$, we have $1\le L_{\N,k}\le k=O(1)$ and $m-k=\Theta(m)$. 
	If $(m-k)^{1/t}<2k=O(1)$, the lower bound follows trivially. 
	Otherwise, by \Cref{thm:det-committee-lower}, we have that
	\[
	\distortion_\N(f,\tr) \ge \frac{(m-k)^{1/t}-k}{L_{\N,k}} = \Omega(m^{1/t}).
	\]
	The matching upper bound follows from \Cref{thm:det}.
\end{proof}

\subsection{Randomized Rules}

The randomized lower bound follows the same two-part structure as the deterministic one. If the final threshold is large, all candidates can again be indistinguishable in the reports. A randomized rule, however, need not omit any candidate with probability one. Instead, we place high utility on a candidate with small marginal inclusion probability. When the final threshold is sufficiently small, we again exploit a large gap between consecutive auxiliary thresholds.

\bigskip
\noindent 
{\bf The last-threshold obstruction.}
Suppose every voter reports all candidates in the last approval set. Since every randomized outcome contains exactly $k$ candidates, the sum of the marginal inclusion probabilities of all candidates is $k$. Hence, some candidate is selected with probability at most $k/m$. Assigning this candidate a large utility while keeping all candidates in the last approval set gives the following bound.

\begin{lemma}\label{lem:rand-last-committee}
	For $1\le k<m$ and any normalized monotone symmetric norm $\snorm{\cdot}$, if $\tau_t>1/m$, then the distortion of every randomized size-$k$ committee rule $f$ using $\tr$ is
	\[
	\distortion_\N(f,\tr) \ge \frac{m\tau_t}{L_{\N,k}+k}.
	\]
\end{lemma}

\begin{proof}
	Consider the input profile in which $\sprofile_{i,t+1}=\candidates$ for every voter $i$. This profile is feasible since the uniform utility vector assigns utility $1/m<\tau_t$ to every candidate. For each candidate $c$, let $p_c$ be its marginal inclusion probability under $f(\sprofile)$. Since every realized committee has size $k$,
	$\sum_{c\in\candidates}p_c = k$.
	Therefore, there exists a candidate $c^*$ with $p_{c^*}\le k/m$.
	
	Choose any $h$ with $1/m<h<\tau_t$, and assign every voter utilities
	\[
	\ut_i(c)=
	\begin{cases}
		h,&\text{ if }c=c^*,\\[3pt]
		\dfrac{1-h}{m-1},&\text{ otherwise}.
	\end{cases}
	\]
	These utilities satisfy the unit-sum assumption. Moreover, since $h>1/m$, we have $(1-h)/(m-1)<1/m<h<\tau_t$, so every candidate remains in the last approval set, and the constructed utility profile induces exactly $\sprofile$.
	
	For any realized committee $\cmt\sim f(\sprofile)$, coordinate-wise,
	\[
	\ut_i^{\cmt} \le
	\frac{1-h}{m-1}\ones_{\cmt} + \left(h-\frac{1-h}{m-1}\right) \ones_{\{c^*\in\cmt\}}\mathbf{e}_{c^*}.
	\]
	By monotonicity, the triangle inequality, positive homogeneity, and normalization,
	\[
	\snorm{\ut_i^{\cmt}} \le \frac{1-h}{m-1}L_{\N,k} + \left(h-\frac{1-h}{m-1}\right) \ones_{\{c^*\in\cmt\}}.
	\]
	Taking expectations and using $p_{c^*}\le k/m$ gives
	\[
	\expected_{\cmt\sim f(\sprofile)} \left[\snorm{\ut_i^{\cmt}}\right]
	\le \frac{1-h}{m-1}L_{\N,k} + \frac{k}{m} \left(h-\frac{1-h}{m-1}\right)
	< \frac{L_{\N,k}+k}{m}.
	\]
	Summing over all voters,
	\[
	\expected_{\cmt\sim f(\sprofile)}[\sw_\N(\cmt,\uprofile)]
	< \frac{n(L_{\N,k}+k)}{m}.
	\]
	On the other hand, any size-$k$ committee $\widehat\cmt$ containing $c^*$ satisfies $\ut_i^{\widehat\cmt}\ge h\mathbf{e}_{c^*}$ coordinate-wise. By monotonicity and normalization, $s_\N^*\ge \sw_\N(\widehat\cmt)\ge nh$. Hence,
	\[
	\distortion_\N(f,\tr)
	>
	\frac{mh}{L_{\N,k}+k}.
	\]
	Taking the supremum as $h$ tends to $\tau_t$ from below proves the claim.
\end{proof}

\bigskip
\noindent
{\bf Large gaps between consecutive thresholds.}
We next adapt the large-gap construction to randomized rules. As before, a large value of $R_j$ identifies an approval set in which utilities can vary substantially without changing the report. We now create sufficiently many candidates in this approval set so that at least one of them has small marginal inclusion probability.

\begin{lemma}\label{lem:rand-large-gap}
	For $1\le k<m$ and any normalized monotone symmetric norm $\snorm{\cdot}$, fix a threshold vector $\tr$ and $j\in[t]$. If $R_j\ge2k$, then the distortion of every randomized size-$k$ committee rule $f$ using $\tr$ is
	\[ 
	\distortion_\N(f,\tr) \ge \frac{R_j}{2(L_{\N,k}+k)}.
	\]
\end{lemma}

\begin{proof}
	Fix $j\in[t]$ with $R_j\ge2k$. Since $\tau_{j-1}\le1$ and $\sigma_j\ge1/(m-k)$, we have $R_j\le m-k$. Since $R_j>1$, we also have $\sigma_{j-1}>\sigma_j\ge1/(m-k)$, and hence $\sigma_{j-1}=\tau_{j-1}$. Therefore,
	\[
	R_j=\frac{\tau_{j-1}}{\sigma_j}.
	\]
	For simplicity, let $q\coloneqq\lfloor R_j/2\rfloor$. Then, $q\le R_j/2$ and $q+1>R_j/2$. Choose a set $A\subseteq\candidates$ with $|A|=q+1$. This is possible since $R_j\le m-k<m$.
	
	We will assign utility $\tau_{j-1}/2$ to one candidate in $A$, utility $\sigma_j$ to the other $q$ candidates in $A$, and split the remaining utility equally among the $m-q-1$ candidates outside $A$. The latter utility value is
	\[
	\frac{1-\tau_{j-1}/2-q\sigma_j}{m-q-1}.
	\]
	It is non-negative because $q\le R_j/2$ and $R_j\sigma_j=\tau_{j-1}$. It is also at most $\sigma_j$, since
	\[
	(m-q-1)\sigma_j
	-\left(1-\frac{\tau_{j-1}}2-q\sigma_j\right)
	=
	\left(m-1+\frac{R_j}{2}\right)\sigma_j-1
	\ge0,
	\]
	where the last inequality follows from $\sigma_j\ge1/(m-k)$ and $m-1+R_j/2\ge m-k$.
	
	Consider the input profile $\sprofile$ in which, for every voter, all candidates in $A$ are reported in approval set $j$, while every candidate outside $A$ is reported in the approval set corresponding to the utility value above. For each $c\in A$, let $p_c$ be its marginal inclusion probability under $f(\sprofile)$. Since every realized committee has size $k$,
	\[
	\sum_{c\in A}p_c = \expected_{\cmt\sim f(\sprofile)}[|\cmt\cap A|] \le k.
	\]
	Hence, there exists $c^*\in A$ with
	\[
	p_{c^*} \le \frac{k}{q+1} < \frac{2k}{R_j}.
	\]
	Now assign every voter utilities
	\[
	\ut_i(c)=
	\begin{cases}
		\dfrac{\tau_{j-1}}2,&\text{ if }c=c^*,\\[6pt]
		\sigma_j,&\text{ if }c\in A\setminus\{c^*\},\\[3pt]
		\dfrac{1-\tau_{j-1}/2-q\sigma_j}{m-q-1},&\text{ if }c\notin A.
	\end{cases}
	\]
	The utilities satisfy the unit-sum assumption by construction. Moreover, since $R_j\ge2k\ge2$,
	\[
	\tau_j
	\le
	\sigma_j
	\le
	\frac{\tau_{j-1}}2
	<
	\tau_{j-1}.
	\]
	Thus, both utility values assigned to candidates in $A$ lie in the interval corresponding to approval set $j$, and the constructed utility profile induces exactly $\sprofile$.
	
	For any realized committee $\cmt\sim f(\sprofile)$, every selected candidate other than possibly $c^*$ has utility at most $\sigma_j$. Hence, coordinate-wise,
	\[
	\ut_i^{\cmt} \le \sigma_j\ones_{\cmt} 
	+ \left(\frac{\tau_{j-1}}2-\sigma_j\right)\ones_{\{c^*\in\cmt\}}\mathbf{e}_{c^*}.
	\]
	By monotonicity, the triangle inequality, positive homogeneity, and normalization,
	\[
	\snorm{\ut_i^{\cmt}} \le \sigma_jL_{\N,k}
	+ \left(\frac{\tau_{j-1}}2-\sigma_j\right) \ones_{\{c^*\in\cmt\}}.
	\]
	Taking expectations and using $p_{c^*}<2k/R_j$ gives
	\[
	\expected_{\cmt\sim\fr(\sprofile)} \left[\snorm{\ut_i^{\cmt}}\right]
	\le \sigma_jL_{\N,k} + \frac{2k}{R_j}\cdot\frac{\tau_{j-1}}2
	= \sigma_j(L_{\N,k}+k).
	\]
	Summing over all voters,
	\[
	\expected_{\cmt\sim\fr(\sprofile)} [\sw_\N(\cmt,\uprofile)]
	\le n\sigma_j(L_{\N,k}+k).
	\]
	On the other hand, any size-$k$ committee $\widehat\cmt$ containing $c^*$ satisfies $\ut_i^{\widehat\cmt}\ge(\tau_{j-1}/2)\mathbf{e}_{c^*}$ coordinate-wise. By monotonicity and normalization,
	\[
	s_\N^* 
	\ge \sw_\N(\widehat\cmt)
	\ge \frac{n\tau_{j-1}}2.
	\]
	Consequently,
	\[
	\distortion_\N(\fr,\tr)
	\ge \frac{\tau_{j-1}}{2\sigma_j(L_{\N,k}+k)}
	= \frac{R_j}{2(L_{\N,k}+k)}.
	\]
	This completes the proof. 
\end{proof}

We can now combine the last-threshold and large-gap bounds to obtain the randomized lower bound.

\begin{theorem}\label{thm:rand-committee-lower}
	For $1\le k<m$, $t\ge1$, and any normalized monotone symmetric norm $\snorm{\cdot}$, the distortion of every randomized size-$k$ committee rule $f$ using an arbitrary threshold vector $\tr$ is
	\[
	\distortion_\N(f,\tr) \ge \max\left\{1,\,\frac{(m-k)^{1/(t+1)}}{2(L_{\N,k}+k)}\right\}.
	\]
\end{theorem}

\begin{proof}
	If $(m-k)^{1/(t+1)}<2k$, then
	\[
	\frac{(m-k)^{1/(t+1)}}{2(L_{\N,k}+k)}
	< \frac{k}{L_{\N,k}+k}
	< 1
	\]
	and the lower bound is trivially $\distortion_\N(\fr,\tr)\ge1$.
	It remains to consider the case $(m-k)^{1/(t+1)}\ge 2k$. 
	We distinguish two cases according to the final threshold.
	\begin{itemize}
		\item
		If $\tau_t>(m-k)^{-t/(t+1)}$, then $\tau_t>1/m$ since $(m-k)^{-t/(t+1)}>1/m$, and \Cref{lem:rand-last-committee} applies. Moreover,
		\[
		m\tau_t 
		> m(m-k)^{-t/(t+1)}
		= \frac{m}{m-k}(m-k)^{1/(t+1)}
		\ge (m-k)^{1/(t+1)}.
		\]
		Hence,
		\[
		\distortion_\N(f,\tr)
		\ge \frac{m\tau_t}{L_{\N,k}+k}
		\ge \frac{(m-k)^{1/(t+1)}}{2(L_{\N,k}+k)}.
		\]
		
		\item
		If $\tau_t\le(m-k)^{-t/(t+1)}$, recall that $\sigma_j=\max\{\tau_j,1/(m-k)\}$ and $R_j=\sigma_{j-1}/\sigma_j$. We have $\sigma_0=1$ and $\sigma_t \le (m-k)^{-t/(t+1)}$  because both $\tau_t$ and $1/(m-k)$ are at most this quantity. Therefore,
		\[
		\prod_{j=1}^tR_j
		= \frac{\sigma_0}{\sigma_t}
		\ge (m-k)^{t/(t+1)}
		\]
		and there exists $j^*\in[t]$ such that
		$R_{j^*} \ge (m-k)^{1/(t+1)} \ge 2k$.
		By \Cref{lem:rand-large-gap},
		\[
		\distortion_\N(f,\tr)
		\ge \frac{R_{j^*}}{2(L_{\N,k}+k)}
		\ge \frac{(m-k)^{1/(t+1)}}{2(L_{\N,k}+k)}.
		\]
	\end{itemize}
	The claim follows by putting everything together. 
\end{proof}

To compare the randomized lower bound with the upper bound from the previous section, we next specialize \Cref{thm:rand-committee-lower} to the utility functions introduced earlier. Substituting $L_{\ell_\infty,k}=1$, $L_{\ell_1,k}=k$, and $L_{\ell_p,k}=k^{1/p}$ gives the following bounds.

\begin{corollary}\label{cor:rand-committee-lower-norms}
	Let $1\le k<m$ and $t\ge1$. The distortion of every randomized size-$k$ committee rule $\fr$ using an arbitrary threshold vector $\tr$ satisfies the following bounds:
	\begin{enumerate}
		\item \textbf{Best-representative utility:}
		\[
		\distortion_{\ell_\infty}(\fr,\tr)
		\ge
		\max\left\{
		1,\,
		\frac{(m-k)^{1/(t+1)}}{2(k+1)}
		\right\}.
		\]
		
		\item \textbf{Additive utility:}
		\[
		\distortion_{\ell_1}(\fr,\tr)
		\ge
		\max\left\{
		1,\,
		\frac{(m-k)^{1/(t+1)}}{4k}
		\right\}.
		\]
		
		\item \textbf{$\ell_p$ utility with $p\ge2$:}
		\[
		\distortion_{\ell_p}(\fr,\tr)
		\ge
		\max\left\{
		1,\,
		\frac{(m-k)^{1/(t+1)}}{2(k^{1/p}+k)}
		\right\}.
		\]
	\end{enumerate}
\end{corollary}

In the single-winner case $k=1$, and more generally for any constant committee size, the lower bound of \Cref{thm:rand-committee-lower} matches the upper bound of \Cref{thm:rand} up to a constant factor. Hence, the optimal randomized distortion is $\Theta(m^{1/(t+1)})$ in these cases.

\begin{corollary}\label{cor:constant-k-rand-tight}
	For any fixed $k=O(1)$ and normalized monotone symmetric norm, the optimal randomized distortion with $t$ thresholds is $\Theta(m^{1/(t+1)})$.
\end{corollary}

\begin{proof}
	For any fixed $k$, we have $1\le L_{\N,k}\le k=O(1)$ and $m-k=\Theta(m)$. Hence, \Cref{thm:rand-committee-lower} gives an $\Omega(m^{1/(t+1)})$ lower bound. The matching upper bound follows from \Cref{thm:rand}.
\end{proof}

\section{Best-Representative Utility}\label{sec:best-rep}
While our upper bounds in \Cref{sec:upper} hold for every norm in our framework, the general lower bounds derived in \Cref{sec:lower} do not necessarily match them when the committee size grows with $m$. Here, we focus on best-representative utility, corresponding to the $\ell_\infty$-norm, which is one of the most widely studied objectives in multiwinner voting. We establish tight deterministic bounds whenever the committee excludes a constant fraction of the candidates and tight randomized bounds for every committee size.

\subsection{Deterministic Rules}

We first consider committees of size $k\le m/2$.

\begin{theorem}\label{thm:lower-for-best}
For $1\le k\le m/2$ and every threshold vector $\tr$, the distortion of every deterministic rule $f$ is
\[
    \distortion_{\ell_\infty}(f,\tr)=\Omega(m^{1/t}).
\]
\end{theorem}

\begin{proof}
If $\tau_t>1/(m-k)$, the distortion is unbounded by
\Cref{lem:unbounded}. Otherwise, let
$\sigma_j=\max\{\tau_j,1/(m-k)\}$ and
$R_j=\sigma_{j-1}/\sigma_j$, as in the proof of
\Cref{thm:det-committee-lower}. Since
$\prod_{j=1}^t R_j=m-k$, there exists $j^*\in[t]$
such that
\[
    R_{j^*}\ge(m-k)^{1/t}\ge\frac12m^{1/t}.
\]
Since the distortion is at least $1$, we may assume that $m^{1/t}$ is sufficiently large, and hence that
$R_{j^*}\ge16$. Let $q=\lfloor R_{j^*}/8\rfloor$. Then $q\ge2$ and $R_{j^*}/16\le q\le R_{j^*}/8$.
Moreover, since $R_{j^*}>1$, we have
$\sigma_{j^*-1}=\tau_{j^*-1}$, and thus
\[
    \tau_{j^*}\le\sigma_{j^*}
    =\frac{\tau_{j^*-1}}{R_{j^*}}
    \le\frac{\tau_{j^*-1}}{16}
    <\frac{\tau_{j^*-1}}4.
\]
For every voter $i$, consider a set
$E_i\subseteq\candidates$ of size $q$ and a candidate
$c_i^*\in E_i$. We define the utility of $i$ as follows:
\[
    \ut_i(c)=
    \begin{cases}
        \tau_{j^*-1}/4, & \text{if } c=c_i^*,\\
        \sigma_{j^*}, & \text{if }  c\in E_i\setminus\{c_i^*\},\\
        \dfrac{1-\tau_{j^*-1}/4-(q-1)\sigma_{j^*}}{m-q}, & \text{if }  c\notin E_i.
    \end{cases}
\]
These utilities satisfy the unit-sum assumption by
construction. They are also non-negative, since
\[
    \frac{\tau_{j^*-1}}4+(q-1)\sigma_{j^*}
    \le\frac{\tau_{j^*-1}}4
    +\frac{R_{j^*}\sigma_{j^*}}8
    =\frac{3\tau_{j^*-1}}8
    \le\frac38.
\]
Moreover, since $R_{j^*}\le m-k\le m$, we have $q\le m/8$. Since the numerator in the utility
for every candidate outside $E_i$ is at most $1$, and $\sigma_{j^*}\ge1/m$, this utility is at most
\[
    \frac{1}{m-q}
    \le\frac{8}{7m}
    \le\frac{8\sigma_{j^*}}7
    <\frac{\tau_{j^*-1}}4.
\]
Consequently, the utility for every candidate is at most $\tau_{j^*-1}/4$.

Since both $\sigma_{j^*}$ and $\tau_{j^*-1}/4$ lie in the same threshold interval, the reported approval sets do not depend on the choice of $c_i^*$. We can therefore choose these candidates after observing the committee selected by the rule. We now distinguish two cases depending on the value
of $kR_{j^*}$, using a different construction of the sets $E_i$ in each case.

\medskip
\noindent 
{\bf Case 1: $kR_{j^*}\le m$.}
Consider $2k$ voters with pairwise disjoint sets $E_i$, which is possible since $2kq\le kR_{j^*}/4\le m$.
Let $\cmt=f(\sprofile)$. Since $\cmt$ contains $k$ candidates, it intersects at most $k$ of the sets $E_i$.
Thus, at least $k$ of these sets are disjoint from $\cmt$. For every voter with $E_i\setminus\cmt\ne\varnothing$, choose $c_i^*\in E_i\setminus\cmt$, and select the remaining favorites arbitrarily. These choices do not affect the reported threshold approvals, so the rule still selects $\cmt$. Since the sets $E_i$ are pairwise disjoint, a committee containing the favorites of $k$ voters whose sets are disjoint from $\cmt$ achieves welfare at least $k\tau_{j^*-1}/4$. Hence, $\sstar_{\ell_\infty}\ge k\tau_{j^*-1}/4$.

By the choice of $c_i^*$, every voter with $E_i\not\subseteq\cmt$ has $c_i^*\notin\cmt$. Such a voter obtains utility at most $\sigma_{j^*}$ from candidates in $E_i$ and at most $8\sigma_{j^*}/7$ from candidates outside $E_i$. Since there are $2k$ voters, their total contribution to the social welfare is at most $2k\cdot 8\sigma_{j^*}/7=16k\sigma_{j^*}/7$. On the other hand, since the sets $E_i$ are pairwise disjoint and each contains $q$ candidates, at most $k/q$ voters have $E_i\subseteq\cmt$. Since the utility for every candidate is at most $\tau_{j^*-1}/4$, each such voter obtains at most this much utility, and their total contribution is at most $k\tau_{j^*-1}/(4q)$. Moreover, since $q\ge R_{j^*}/16$ and $R_{j^*}=\tau_{j^*-1}/\sigma_{j^*}$, we have $\tau_{j^*-1}/(4q)\le4\sigma_{j^*}$. Therefore,
\[
    \sw_{\ell_\infty}(\cmt)
    \le\frac{16}{7}k\sigma_{j^*}
    +\frac{k\tau_{j^*-1}}{4q}
    \le\left(\frac{16}{7}+4\right)k\sigma_{j^*}
    <7k\sigma_{j^*}.
\]
Putting everything together, the distortion is at least $\tau_{j^*-1}/(28\sigma_{j^*})=R_{j^*}/28 = \Omega(m^{1/t})$.

\medskip
\noindent 
{\bf Case 2: $kR_{j^*}>m$.}
Consider an instance with one voter for every size-$q$ subset $E_i$ of $\candidates$, so that $n=\binom{m}{q}$. Let $\cmt=f(\sprofile)$. Since $k\le m/2$, there exists a committee $\cmt'\subseteq\candidates\setminus\cmt$ of size $k$. For every voter with $E_i\cap\cmt'\ne\varnothing$, choose $c_i^*\in E_i\cap\cmt'$. For every remaining voter with $E_i\setminus\cmt\ne\varnothing$, choose $c_i^*\in E_i\setminus\cmt$, and select the remaining favorites arbitrarily. These choices do not affect the reported threshold approvals, so the rule still selects $\cmt$.

We first lower-bound the optimal welfare by considering $\cmt'$. Since there is one voter for every size-$q$ subset of $\candidates$, the fraction of voters whose sets intersect $\cmt'$ is
\[
    1-\frac{\binom{m-k}{q}}{\binom{m}{q}}
    \ge 1-\left(1-\frac{k}{m}\right)^q
    \ge 1-e^{-kq/m}
    \ge 1-e^{-1/16},
\]
where the last inequality follows from $q\ge R_{j^*}/16$ and $kR_{j^*}>m$. By the choice of $c_i^*$, the favorite candidate of each such voter belongs to $\cmt'$ and her utility is $\tau_{j^*-1}/4$. Therefore, 
\[
    \sstar_{\ell_\infty}\ge
    (1-e^{-1/16})\frac{n\tau_{j^*-1}}{4}.
\]
We next upper-bound the welfare of $\cmt$. By the choice of $c_i^*$, every voter with $E_i\not\subseteq\cmt$ has $c_i^*\notin\cmt$, and thus obtains utility at most $8\sigma_{j^*}/7$. Since there are $n$ voters, their total contribution to the social welfare is at most $n \cdot 8\sigma_{j^*}/7$. On the other hand, since there is one voter for every size-$q$ subset of $\candidates$, the number of voters with $E_i\subseteq\cmt$ is $\binom{k}{q}$. Their fraction is therefore at most
\[
    \frac{\binom{k}{q}}{\binom{m}{q}}
    \le \left(\frac{k}{m}\right)^q
    \le 2^{-q},
\]
where the last inequality follows from $k\le m/2$. Since the utility for any candidate is at most $\tau_{j^*-1}/4$, each such voter obtains at most this much utility, and their total contribution is at most $n\cdot 2^{-q}\tau_{j^*-1}/4$. Moreover, since $q\ge R_{j^*}/16$ and $R_{j^*}=\tau_{j^*-1}/\sigma_{j^*}$, we have $\tau_{j^*-1}/4\le4q\sigma_{j^*}$, and $q2^{-q}\le1/2$ for $q\ge2$. Therefore,
\[
\begin{aligned}
    \sw_{\ell_\infty}(\cmt)
    &\le \frac{8}{7}n\sigma_{j^*}
    +n2^{-q}\frac{\tau_{j^*-1}}{4}\\
    &\le \left(\frac87+4q2^{-q}\right)n\sigma_{j^*}\\
    &<4n\sigma_{j^*}.
\end{aligned}
\]
Putting everything together, the distortion is at least
\[
    \frac{(1-e^{-1/16})n\tau_{j^*-1}/4}
    {4n\sigma_{j^*}}
    =\frac{1-e^{-1/16}}{16}R_{j^*}
    =\Omega(m^{1/t}). \qedhere
\]
\end{proof}

We next consider the case $m/2<k<m$. 

\begin{theorem}\label{thm:best-large-committees}
For $m/2<k<m$ and every threshold vector $\tr$, the distortion of every deterministic rule $f$ is
\[
    \distortion_{\ell_\infty}(f,\tr)=\Omega((m-k)^{1/t}).
\]
\end{theorem}

\begin{proof}
Let $D=\distortion_{\ell_\infty}(f,\tr)$. By \Cref{lem:unbounded}, the distortion is unbounded if $\tau_t>1/(m-k)$. Hence, we may assume that $\tau_t\le1/(m-k)$. For every $j\in\{0,\ldots,t\}$, let $\sigma_j=\max\{\tau_j,1/m\}$.

We first show that $\sigma_t\le D/m$. If $\tau_t\le1/m$, then, since $D \geq 1$, we have $\sigma_t=1/m\le D/m$. Otherwise, since $\tau_t>1/m$, consider an instance with a single voter $i$ who reports all candidates in the last approval set. Let $\cmt=f(\sprofile)$ and choose a candidate $x\notin\cmt$. 
For any $u\in(1/m,\tau_t)$, voter $i$ has utility $u$ for $x$ and utility $(1-u)/(m-1)$ to every other candidate. Since $(1-u)/(m-1)<1/m<\tau_t$, these utilities induce the same report, and thus the chosen committee must be $\cmt$. Since the welfare of $W$ is $(1-u)/(m-1)$ and the optimal welfare is $u$, taking $u\to\tau_t$ and using $\tau_t>1/m$, we obtain
\[
    D\ge\frac{(m-1)\tau_t}{1-\tau_t}\ge m\tau_t.
\] 
Therefore, $\sigma_t\le D/m$. Combining this with $\tau_t\le1/(m-k)$ and $1/m\le1/(m-k)$, we conclude that
\begin{equation}\label{eq:large-committee-last}
    \sigma_t\le\min\left\{\frac Dm,\frac1{m-k}\right\}.
\end{equation}
We next show that, for every $j\in[t]$,
\begin{equation}\label{eq:large-committee-step}
    \sigma_{j-1}\le640D\max\left\{\sigma_j,\frac{m\sigma_j^2}{m-k}\right\}.
\end{equation}
As in previous proofs, let $R_j=\sigma_{j-1}/\sigma_j$. If $R_j\le640D$, the inequality follows directly from the definition of $R_j$. Hence, we may assume that $R_j>640D\ge640$. Since $\sigma_j\ge1/m$, this implies that $\sigma_{j-1} = R_j \sigma_j >1/m$, and thus $\sigma_{j-1}=\tau_{j-1}$. Moreover, $\sigma_j=\tau_{j-1}/R_j<1/640$.
Using all these, we have
\[
    \tau_j\le\sigma_j
    =\frac{\tau_{j-1}}{R_j}
    <\frac{\tau_{j-1}}{16}
    <\frac{\tau_{j-1}}4
    <\tau_{j-1}.
\]
Let $q=\lfloor1/(8\sigma_j)\rfloor$. Since $1/m\le\sigma_j<1/16$, we have $2\le q\le m/8$. Moreover, since $1/(8\sigma_j)\ge2$, we have $q\ge1/(16\sigma_j)$.

Consider an instance with one voter for every size-$q$ subset $E_i$ of $\candidates$, so that $n=\binom mq$. For each voter $i$, we consider two possible utility functions: 
\[
    \ut_i^{(1)}(c)=
    \begin{cases}
        \tau_{j-1}/4, & \text{ if } c=c_i^*,\\[2mm]
        \sigma_j, & \text{ if }  c\in E_i\setminus\{c_i^*\},\\[2mm]
        \dfrac{1-\tau_{j-1}/4-(q-1)\sigma_j}{m-q}, & \text{ if }  c\notin E_i.
    \end{cases}
\]
and 
\[
    \ut_i^{(2)}(c)=
    \begin{cases}
        \dfrac{\tau_{j-1}/4+(q-1)\sigma_j}{q}, & \text{ if }  c\in E_i,\\[4mm]
        \dfrac{1-\tau_{j-1}/4-(q-1)\sigma_j}{m-q}, & \text{ if }  c\notin E_i.
    \end{cases}
\]
Both utility functions satisfy the unit-sum assumption by construction. They are also non-negative, since
\[
    \frac{\tau_{j-1}}4+(q-1)\sigma_j
    \le\frac{\tau_{j-1}}4+\frac18
    \le\frac38.
\]
Moreover, since $q\le m/8$ and $\sigma_j\ge1/m$, for every candidate $c\notin E_i$ we have
\[
    \ut_i^{(1)}(c)=\ut_i^{(2)}(c)
    \le\frac1{m-q}
    \le\frac8{7m}
    \le\frac{8\sigma_j}7.
\]
For every candidate $c\in E_i$, the utility $\ut_i^{(1)}(c)$ lies between $\sigma_j$ and $\tau_{j-1}/4$. The same holds for $\ut_i^{(2)}(c)$, since it is the average of one value $\tau_{j-1}/4$ and $q-1$ values $\sigma_j$. Since
$\tau_j\le\sigma_j<\tau_{j-1}/4<\tau_{j-1}$,
the utilities of voter $i$ for all candidates in $E_i$ belong to the same threshold interval under both functions. Together with the fact that the utilities outside $E_i$ are identical, this implies that both utility functions induce the same threshold approvals, independently of the choice of $c_i^*$.

Let $\cmt=f(\sprofile)$. For every voter $i$ with $E_i\setminus\cmt\ne\varnothing$, let $\ut_i=\ut_i^{(1)}$, choosing $c_i^*\in E_i\setminus\cmt$. For every voter with $E_i\subseteq\cmt$, let $\ut_i=\ut_i^{(2)}$. Since both utility functions induce the same threshold approvals, these choices do not affect the selected committee $\cmt$.
We first upper-bound the welfare of $\cmt$.
\begin{itemize}
    \item If $E_i\not\subseteq\cmt$, then by the choice of $c_i^*$ in $\ut_i^{(1)}$, we have $c_i^*\notin\cmt$. Thus, voter $i$ has utility at most $\sigma_j$ for every candidate in $E_i\cap\cmt$ and at most $8\sigma_j/7$ for every candidate in $\cmt\setminus E_i$. Therefore, her utility for $\cmt$ is at most $8\sigma_j/7$.

    \item If $E_i\subseteq\cmt$, then since $q\ge1/(16\sigma_j)$ and $\tau_{j-1}\le1$, for every candidate $c\in E_i$ we have
    \[
        \ut_i^{(2)}(c)
        =\frac{\tau_{j-1}/4+(q-1)\sigma_j}{q}
        \le\sigma_j+\frac{\tau_{j-1}}{4q}
        \le\sigma_j+4\tau_{j-1}\sigma_j
        \le5\sigma_j.
    \]
    Moreover, the utility of $i$ for every candidate outside $E_i$ is at most $8\sigma_j/7<5\sigma_j$. Hence, her utility for $\cmt$ is at most $5\sigma_j$.
\end{itemize}
Since $8\sigma_j/7<5\sigma_j$, every voter obtains utility at most $5\sigma_j$ for $\cmt$. 
Summing over all voters, we obtain
\[
    \sw_{\ell_\infty}(\cmt)\le5n\sigma_j.
\]

We next lower-bound the optimal welfare by considering a committee $\cmt'$ of size $k$ containing all candidates outside $\cmt$; note that such a committee must exist since $m-k<k$. We have two types of voters:
\begin{itemize}
    \item If $E_i\subseteq\cmt$, then voter $i$ has utility function $\ut_i^{(2)}$. Since there is one voter for every size-$q$ subset of $\candidates$, there are exactly $\binom{k}{q}$ such voters. We do not need to account their contribution to the welfare of $\cmt'$.

    \item If $E_i\not\subseteq\cmt$, then voter $i$ has utility function $\ut_i^{(1)}$, with $c_i^*\in E_i\setminus\cmt$. Since $\cmt'$ contains all candidates outside $\cmt$, we have $c_i^*\in\cmt'$, and thus voter $i$ obtains utility at least $\tau_{j-1}/4$ for $\cmt'$. There are $n-\binom{k}{q}$ such voters.
\end{itemize}
The fraction of voters who obtain utility at least $\tau_{j-1}/4$ for $\cmt'$ is at least $1-\binom{k}{q}/\binom{m}{q}$. Using $1-x\le e^{-x}$ and $q\ge1/(16\sigma_j)$, we have
\[
    1-\frac{\binom{k}{q}}{\binom mq}
    =1-\prod_{\ell=0}^{q-1}\frac{k-\ell}{m-\ell}
    \ge1-\left(\frac km\right)^q
    =1-\left(1-\frac{m-k}{m}\right)^q
    \ge1-e^{-(m-k)q/m}
    \ge1-e^{-(m-k)/(16m\sigma_j)}.
\]
Since each voter of the second type obtains utility at least $\tau_{j-1}/4$ for $\cmt'$, we conclude that
\[
    \sstar_{\ell_\infty}
    \ge\frac{n\tau_{j-1}}4
    \left(1-e^{-(m-k)/(16m\sigma_j)}\right).
\]

Combining the bounds on the optimal welfare and the welfare of $\cmt$, and using $\tau_{j-1}=R_j\sigma_j$, we obtain
\[
\begin{aligned}
    D
    &\ge\frac{\tau_{j-1}}{20\sigma_j}
    \left(1-e^{-(m-k)/(16m\sigma_j)}\right)\\
    &=\frac{R_j}{20}
    \left(1-e^{-(m-k)/(16m\sigma_j)}\right)\\
    &\ge\frac{R_j}{40}
    \min\left\{1,\frac{m-k}{16m\sigma_j}\right\}\\
    &\ge\frac{R_j}{640}
    \min\left\{1,\frac{m-k}{m\sigma_j}\right\},
\end{aligned}
\]
where the third inequality follows from $1-e^{-x}\ge\frac12\min\{1,x\}$ for every $x\ge0$. Since $R_j=\sigma_{j-1}/\sigma_j$, rearranging the last inequality yields \eqref{eq:large-committee-step}.

It remains to combine \eqref{eq:large-committee-last} and \eqref{eq:large-committee-step} to obtain the desired lower bound. Let $B=640D$ and suppose towards a contradiction that $B^t<m-k$. We show by backward induction that $\sigma_j\le B^{t-j}\sigma_t$ for every $j\in\{0,\ldots,t\}$. The claim clearly holds for $j=t$. Suppose that it holds for some $j\ge1$. By \eqref{eq:large-committee-last}, we have $\sigma_t\le D/m\le B/m$. Using $B\ge1$ and $j\ge1$, we have that
\[
    \sigma_j
    \le B^{t-j}\sigma_t
    \le\frac{B^{t-j+1}}m
    \le\frac{B^t}m
    <\frac{m-k}{m}
\]
Since $\sigma_j<(m-k)/m$, we further have $m\sigma_j^2/(m-k)<\sigma_j$. Therefore, by \eqref{eq:large-committee-step}, $\sigma_{j-1} \le B\sigma_j \le B^{t-j+1}\sigma_t$, completing the induction step. The bound for $j=0$, together with $\sigma_0=1$ and $\sigma_t\le1/(m-k)$ from \eqref{eq:large-committee-last}, implies that
\[
    1=\sigma_0
    \le B^t\sigma_t
    \le\frac{B^t}{m-k}
    <1,
\]
a contradiction. Hence, $B^t\ge m-k$, which, by the definition of $B$, implies
\[
    D\ge\frac1{640}(m-k)^{1/t}
    =\Omega((m-k)^{1/t}). \qedhere
\]
\end{proof}

Combining the two lower bounds with the deterministic upper bound from \Cref{thm:det}, we obtain the following characterization for best-representative utility.

\begin{corollary}\label{cor:tight-best-rep}
When $m-k=\Omega(m)$, the optimal deterministic distortion for best-representative utility with $t$ thresholds is $\Theta(m^{1/t})$.
\end{corollary}

\subsection{Randomized Rules}

We now turn to randomized rules and establish tight bounds for every committee size.

\begin{theorem}\label{thm:randomized-best}
For $1\le k\le m$, the optimal randomized distortion for best-representative utility with $t$ thresholds is
\[
    \Theta\left(\min\left\{m^{1/(t+1)},\frac mk\right\}\right).
\]
\end{theorem}

\begin{proof}
The upper bound follows from \Cref{thm:rand}. For the lower bound, consider any randomized rule $f$ using threshold vector $\tr$, and let $a=m^{1/(t+1)}$. If $a<64$ or $k>m/2$, the lower bound follows from the fact that distortion is at least $1$, since $\min\{a,m/k\}\le64$. Hence, assume that $a\ge64$ and $k\le m/2$.

We distinguish two cases depending on the value of the last threshold $\tau_t$.

\medskip
\noindent
{\bf Case 1: $\tau_t>a^{-t}$.}
Since $\tau_t>a^{-t}>1/m$, consider an instance with a single voter $i$ reporting all candidates in the last approval set. Let $\cmt=f(\sprofile)$ be the random committee selected by the rule, and let $p_c$ denote the probability that candidate $c$ is included in $\cmt$. Since the rule always selects a committee of size $k$, we have $\sum_{c\in\candidates}p_c=k$, and thus there exists a candidate $x$ with $p_x\le k/m$.

For any $u\in(1/m,\tau_t)$, we define the utility of $i$ to be $u$ for $x$, and $(1-u)/(m-1)$ for every other candidate. Since $(1-u)/(m-1)<1/m<\tau_t$, these utilities induce the same report, and thus the probabilities $p_c$ remain unchanged. The voter obtains utility $u$ if $x$ is selected, and utility $(1-u)/(m-1)$ otherwise. Therefore, since $u>1/m$ and $p_x\le k/m$, the expected welfare of the rule is at most
\[
    \expected[\sw_{\ell_\infty}(\cmt)] \le \frac{1-u}{m-1}+p_xu \le\frac1m+\frac km u.
\]
The optimal committee contains $x$ and achieves welfare $u$. 
Taking $u\to\tau_t$, the distortion is at least
\[
    \frac{m\tau_t}{1+k\tau_t}
    \ge\frac12\min\left\{m\tau_t,\frac mk\right\}
    \ge\frac12\min\left\{a,\frac mk\right\}.
\]
Here, the second inequality follows from $1+k\tau_t\le2\max\{1,k\tau_t\}$, and the last inequality follows from $\tau_t>a^{-t}$ and $m=a^{t+1}$.

\medskip
\noindent
{\bf Case 2: $\tau_t\le a^{-t}$.}
For every $j\in\{0,\ldots,t\}$, let $\sigma_j=\max\{\tau_j,1/(m-1)\}$, and $R_j=\sigma_{j-1}/\sigma_j$ for every $j\in[t]$. Since $1/(m-1)\le2a^{-t}$, we have $\sigma_t\le2a^{-t}$. Therefore,
\[
    \prod_{j=1}^t R_j
    =\frac{\sigma_0}{\sigma_t}
    =\frac1{\sigma_t}
    \ge\frac{a^t}{2},
\]
and thus there must exist $j\in[t]$ such that $R_j\ge a/2\ge32$. 
Since $R_j>1$ and $\sigma_j\ge1/(m-1)$, we have $\sigma_{j-1}>1/(m-1)$, and thus $\sigma_{j-1}=\tau_{j-1}$. Putting these together, we have
\[
    \tau_j\le\sigma_j
    =\frac{\tau_{j-1}}{R_j}
    \le\frac{\tau_{j-1}}{32}
    \le\frac1{32}.
\]
Let $q=\lfloor1/(4\sigma_j)\rfloor$. Since $1/(m-1)\le\sigma_j\le1/32$, we have $8\le q\le(m-1)/4$. Moreover, since $1/(4\sigma_j)\ge8$ and $\tau_{j-1}=R_j\sigma_j\le1$, we have $q\ge\frac1{8\sigma_j}\ge\frac{R_j}{8}$.

Now let $g=\min\{k,\lfloor m/q\rfloor\}$ and consider an instance with $g$ voters that have pairwise disjoint sets $E_1,\ldots,E_g$, each containing $q$ candidates. Such sets exist since $gq\le m$. 
Every voter $i\in[g]$ reports $E_i$ as her $j$-th approval set and $\candidates\setminus E_i$ as her last approval set.
Let $\cmt=f(\sprofile)$ and, for every candidate $c$, let $p_c$ be the probability that $c$ is included in $\cmt$. For every voter $i$, choose a candidate $x_i\in E_i$ with minimum inclusion probability among the candidates in $E_i$. We define her utility function as follows:
\[
    \ut_i(c)=
    \begin{cases}
        \tau_{j-1}/2, & \text{if } c=x_i,\\
        \dfrac{1-\tau_{j-1}/2}{q-1}, & \text{if }  c\in E_i\setminus\{x_i\},\\
        0, & \text{if }  c\notin E_i.
    \end{cases}
\]
These utilities are non-negative and satisfy the unit-sum assumption by construction. We next verify that they induce the reported threshold approvals. Since $\tau_{j-1}\le1$ and $q\le1/(4\sigma_j)$, we have
\[
    \frac{1-\tau_{j-1}/2}{q-1}
    \ge\frac{1/2}{q-1}
    \ge2\sigma_j.
\]
On the other hand, since $\sigma_j\le1/32$ and $q=\lfloor1/(4\sigma_j)\rfloor$, we have
\[
    q-1\ge\frac1{4\sigma_j}-2
    \ge\frac3{16\sigma_j},
\]
and therefore
\[
    \frac{1-\tau_{j-1}/2}{q-1}
    \le\frac{16\sigma_j}{3}
    <6\sigma_j.
\]
Combining these inequalities with $R_j\ge32$, we obtain
\[
    \tau_j\le\sigma_j
    <2\sigma_j
    \le\frac{1-\tau_{j-1}/2}{q-1}
    <6\sigma_j
    <\frac{\tau_{j-1}}2
    <\tau_{j-1}.
\]
Thus, these utilities induce the reported approval sets, with $E_i$ as the $j$-th approval set and $\candidates\setminus E_i$ as the last approval set. In particular, choosing the candidates $x_i$ after fixing the rule's lottery does not affect the reported threshold approvals or the inclusion probabilities $p_c$.

We first upper-bound the expected welfare of the rule. Since $x_i$ has minimum inclusion probability in $E_i$, we have $p_{x_i}\le\frac1q\sum_{c\in E_i}p_c.$    
As the sets $E_i$ are pairwise disjoint and the rule always selects $k$ candidates, summing over all voters, we obtain
\[
    \sum_{i=1}^g p_{x_i}
    \le\frac1q\sum_{c\in\bigcup_{i=1}^g E_i}p_c
    \le\frac kq.
\]
For every voter $i$, her utility for $x_i$ is $\tau_{j-1}/2$, while her utility for every other candidate is at most $6\sigma_j$. Hence, her expected utility for $\cmt$ is at most $6\sigma_j+p_{x_i}\tau_{j-1}/2$. Summing over all voters, we obtain
\[
    \expected[\sw_{\ell_\infty}(\cmt)]
    \le6g\sigma_j +\frac{\tau_{j-1}}2\sum_{i=1}^g p_{x_i}
    \le6g\sigma_j+\frac{k\tau_{j-1}}{2q}.
\]
We next lower-bound the optimal welfare. Since the sets $E_i$ are pairwise disjoint, the candidates $x_1,\ldots,x_g$ are distinct. Moreover, $g\le k$, and thus there exists a committee of size $k$ containing all these candidates. Every voter $i$ obtains utility $\tau_{j-1}/2$ for this committee, implying
\[
    \sstar_{\ell_\infty}\ge\frac{g\tau_{j-1}}2.
\]
Combining the bounds, and using $\tau_{j-1}=R_j\sigma_j$, the distortion is at least
\[
    \frac{g\tau_{j-1}/2}
    {6g\sigma_j+k\tau_{j-1}/(2q)}
    =\frac1{12/R_j+k/(gq)}.
\]
We distinguish two cases depending on the value of $g$.
\begin{itemize}
    \item If $g=k$, then $k/(gq)=1/q\le8/R_j$, and thus the distortion is at least
    \[
        \frac1{12/R_j+1/q}
        \ge\frac{R_j}{20}.
    \]

    \item If $g<k$, then $g=\lfloor m/q\rfloor$. Since $q\le(m-1)/4$, we have $gq\ge m-q\ge m/2$, and therefore
    \[
        \frac{12}{R_j}+\frac{k}{gq}
        \le\frac{12}{R_j}+\frac{2k}{m}
        \le14\max\left\{\frac1{R_j},\frac km\right\}.
    \]
    Consequently, the distortion is at least
    \[
        \frac1{14}
        \min\left\{R_j,\frac mk\right\}.
    \]
\end{itemize}
In either case, since $R_j\ge a/2$, we conclude that the distortion is at least
\[
    \frac1{20}\min\left\{R_j,\frac mk\right\}
    \ge\frac1{40}\min\left\{a,\frac mk\right\} 
    = \Omega(\min\{m^{1/(t+1)},m/k\}).
\]
This completes the proof.
\end{proof}

\section{Conclusion and Open problems} 
We studied multiwinner voting when voters partition the candidates into threshold approval sets. For a broad class of committee utilities defined through normalized monotone symmetric norms, we designed simple deterministic and randomized rules whose distortion improves as more thresholds are available. Our results show that logarithmically many thresholds suffice for constant distortion, while meaningful guarantees are possible even with a single threshold. We also established lower bounds that show how the achievable distortion depends on the number of candidates, the committee size, and the chosen norm. For single-winner voting and, more generally, for committees of constant size, our upper and lower bounds match up to constant factors, yielding tight distortion bounds of $\Theta(m^{1/t})$ for deterministic rules and $\Theta(m^{1/(t+1)})$ for randomized rules. For best-representative utility, we further established tight randomized bounds for every committee size and tight deterministic bounds whenever $m-k=\Omega(m)$.

Several questions remain open. For general normalized monotone symmetric norms, closing the gap between our upper and lower bounds for growing committee sizes would provide a better understanding of how the committee size and the chosen norm affect distortion. For best-representative utility, determining the optimal deterministic distortion when $1\le m-k=o(m)$ remains open. Another direction concerns computation: maximizing virtual welfare may be computationally difficult for some committee-utility functions, making it natural to ask whether comparable distortion guarantees can be achieved by polynomial-time rules. Finally, it would be interesting to study whether non-geometric, adaptive, or voter-specific thresholds can improve the tradeoff between the amount of elicited information and distortion.

\subsubsection*{AI Use Disclosure}
ChatGPT-5.6 Sol was used to assist with the writing and revision of the manuscript. The authors have reviewed the resulting material and take full responsibility for the content and correctness of the paper.

\bibliographystyle{plainnat} 
\bibliography{references}


\end{document}